\documentclass[11pt]{article}
\usepackage[in]{fullpage}%

\usepackage{graphicx}%
\usepackage{amsmath}%
\usepackage{color}%
\usepackage{xspace}%
\usepackage{paralist}
\usepackage{mleftright}%
\usepackage{amssymb}%
\usepackage{xcolor}%
\usepackage{caption}%
\usepackage{stmaryrd}%

 \usepackage[amsmath,thmmarks]{ntheorem}%
 \theoremseparator{.}%

\usepackage{hyperref}%
\hypersetup{%
   breaklinks,%
   colorlinks=true,%
   linkcolor=[rgb]{0.45,0.0,0.0},%
   citecolor=[rgb]{0,0,0.45}
}

\numberwithin{figure}{section}%
\numberwithin{table}{section}%
\numberwithin{equation}{section}%

\renewcommand{\Re}{\mathbb{R}}%

   \theoremstyle{plain}%
  \newtheorem{theorem}{Theorem}[section]
  
  \newtheorem{lemma}[theorem]{Lemma}%
  \newtheorem{claim}[theorem]{Claim}%

   \newtheorem{fact}[theorem]{Fact}%

  \theoremstyle{plain}%
  \theoremheaderfont{\sf}%
  \theorembodyfont{\upshape}%
  \newtheorem{definition}[theorem]{Definition}
  
  \newtheorem*{remark:unnumbered}[theorem]{Remark}%
  \newtheorem{remark}[theorem]{Remark}%

\newcommand{\HLinkShort}[2]{\hyperref[#2]{#1\ref*{#2}}}
\newcommand{\HLink}[2]{\hyperref[#2]{#1~\ref*{#2}}}
\newcommand{\HLinkPage}[2]{\hyperref[#2]{#1~\ref*{#2}%
      $_\text{p\pageref{#2}}$}}
\newcommand{\HLinkPageOnly}[1]{\hyperref[#1]{Page~\refpage*{#1}%
      $_\text{p\pageref{#1}}$}}

\newcommand{\HLinkSuffix}[3]{\hyperref[#2]{#1\ref*{#2}{#3}}}
\newcommand{\HLinkPageSuffix}[3]{\hyperref[#2]{#1\ref*{#2}%
      #3$_\text{p\pageref{#2}}$}}

\newcommand{\deflab}[1]{\label{def:#1}}
\newcommand{\defref}[1]{\HLink{Definition}{def:#1}}%

\newcommand{\claimlab}[1]{\label{claim:#1}}
\newcommand{\claimref}[1]{\HLink{Claim}{claim:#1}}%
      
\newcommand{\figlab}[1]{\label{fig:#1}}
\newcommand{\figref}[1]{\HLink{Figure}{fig:#1}}

\newcommand{\remlab}[1]{\label{rem:#1}}

\providecommand{\lemlab}[1]{\label{xlemma:#1}}
\renewcommand{\lemlab}[1]{\label{xlemma:#1}}
\newcommand{\lemref}[1]{\HLink{Lemma}{xlemma:#1}}%

\newcommand{\factlab}[1]{\label{fact:#1}}
\newcommand{\factref}[1]{\HLink{Fact}{fact:#1}}%

\newcommand{\thmlab}[1]{{\label{theo:#1}}}
\newcommand{\thmref}[1]{\HLink{Theorem}{theo:#1}}

\newcommand{\seclab}[1]{\label{sec:#1}}

\newcommand{\apndlab}[1]{\label{apnd:#1}}
\newcommand{\apndref}[1]{\HLink{Appendix}{apnd:#1}}

\providecommand{\eqlab}[1]{}%
\renewcommand{\eqlab}[1]{\label{equation:#1}}

\newcommand{\myqedsymbol}{\rule{2mm}{2mm}}

\newcommand{\myparagraph}[1]{\bigskip\noindent{\textbf{#1}}}

\theoremheaderfont{\em}%
\theorembodyfont{\upshape}%
\theoremstyle{nonumberplain} \theoremseparator{}
\theoremsymbol{\myqedsymbol} \newtheorem{proof}{Proof:}
  
\DefineNamedColor{named}{RedViolet} {cmyk}{0.07,0.90,0,0.34}

\newcommand{\VorC}{\mathcal{V}}%
\newcommand{\VorX}[2][\!]{\VorC\pth{#2}}%
\newcommand{\SVorC}{\mathcal{SV}}%
\newcommand{\SVorX}[2][\!]{\SVorC\pth{#2}}%
\newcommand{\WVorC}{\mathcal{WV}}%
\newcommand{\WVorX}[2][\!]{\WVorC\pth{#2}}%

\newcommand{\half}{\mathsf{H}}

\newcommand{\dist}[2]{\left\| #1 - #2 \right\| }

\newcommand{\pth}[1]{\mleft({#1}\mright)}

\newcommand{\pbrcx}[1]{\left[ {#1} \right]}
\newcommand{\Ex}[2][\!]{\mathop{\mathbf{E}}#1\pbrcx{#2}}

\newcommand{\Prob}[1]{\mathop{\mathbf{Pr}}\!\pbrcx{#1}}
\newcommand{\sep}[1]{\,\left|\, {#1} \bigr.\right.}

\newcommand{\SiteSet}{{S}}

\newcommand{\Stretch}{t}

\definecolor{blue25}{rgb}{0,0,0.55}%

\newcommand{\VorCell}[3][\!]{\VorC_\mathrm{cell}\pth{#2, #3}}
\newcommand{\SVorCell}[3][\!]{\SVorC_\mathrm{cell}\pth{#2, #3}}

\newcommand{\remove}[1]{}
\newcommand{\eps}{{\varepsilon}}%

\newcommand{\atgen}{\symbol{'100}}
\newcommand{\BenThanks}[1]{\thanks{Department of Computer Science;
      University of Texas at Dallas; Richardson, TX 75080, USA; 
      {\tt benjamin.raichel\atgen{}utdallas.edu}; {\tt
         \url{http://utdallas.edu/\string~benjamin.raichel}.} #1}}
\newcommand{\ChenglinThanks}[1]{\thanks{Department of Computer Science;
      University of Texas at Dallas; Richardson, TX 75080, USA; 
      {\tt cxf160130\atgen{}utdallas.edu}. #1}}

\usepackage[utf8]{inputenc}
\usepackage{algorithm}
\usepackage[noend]{algpseudocode}

\title{Linear Expected Complexity for Directional and Multiplicative Voronoi Diagrams}%

\begin{document}
\author{%
Chenglin Fan\ChenglinThanks{Work on this paper was partially
      supported by NSF CRII Award 1566137 and CAREER Award 1750780.}
\and
Benjamin Raichel\BenThanks{Work on this paper was partially
      supported by NSF CRII Award 1566137 and CAREER Award 1750780.}
}
\date{}

\maketitle

\begin{abstract}
While the standard unweighted Voronoi diagram in the plane has linear worst-case complexity, many of its natural generalizations do not. 
This paper considers two such previously studied generalizations, namely multiplicative and semi Voronoi diagrams. These diagrams both have quadratic worst-case complexity, though here we show that their expected complexity is linear for certain natural randomized inputs.
Specifically, we argue that the expected complexity is linear for: (1) semi Voronoi diagrams when the visible direction is randomly sampled, and (2) for multiplicative diagrams when either weights are sampled from a constant-sized set, or the more challenging case when weights are arbitrary but locations are sampled from a square.
\end{abstract}

\newpage
\pagenumbering{arabic}

\section{Introduction}
Given a set of point sites in the plane, the Voronoi diagram is the corresponding partition of the plane into cells, where each cell is the locus of points in the plane sharing the same closest site. 
This fundamental structure has a wide variety of applications.  When coupled with a point location data structure, it can be used to quickly answer nearest neighbor queries.
Other applications include robot motion planning, modeling natural processes in areas such as biology, chemistry, and physics, and moreover, the dual of the Voronoi diagram is the well known Delaunay triangulation.  
See the book \cite{Franz2013} for an extensive coverage of the topic.

Many of the varied applications of Voronoi diagrams require generalizing its definition in one way or another, 
such as adding weights or otherwise altering the distance function, moving to higher dimensions, or considering alternative types of sites. 
While some of these generalizations retain the highly desirable linear worst-case complexity of the standard Voronoi diagram, many others unfortunately have quadratic worst-case complexity or more.

Within the field of Computational Geometry, particularly in recent years, there have been a number of works analyzing the expected complexity of various geometric structures when the input is assumed to have some form of randomness. 
Here we continue this line of work, by studying the expected complexity of two previously considered Voronoi diagram variants.  

\myparagraph{Directional and weighted Voronoi diagrams.}
A natural generalization to consider is when each site is only visible to some subset of the plane.
Here we are interested in the so called visual restriction Voronoi diagram (VRVD)~\cite{AurenhammerSXZ14,FanLWZ14}, 
where a given site $p$ is only visible to the subset of points contained in some cone with base point $p$ and angle $\alpha_p$. 
These diagrams model scenarios where the site has a restricted field of view, such as may be the case with various optical sensors or human vision.
For example, in a football game, each player has their own field of view at any given time, and the location of the ball in the VRVD tells us which player is the closest to the ball among those who can see it.
When the visible region for each site is a half-plane whose boundary passes through the site (i.e., a VRVD where $\alpha_p=\pi/2$ for each site $p$), 
such diagrams are called semi Voronoi diagrams~\cite{ChengLX10}.
Just like general VRVD's, semi Voronoi diagrams have $\Theta(n^2)$ worst-case complexity.  
Our expected analysis is shown for the semi Voronoi diagram case, however, we remark a similar analysis implies the same bounds hold more generally for VRVD's. 

The other generalization we consider is when sites have weights.  
There are many natural ways to incorporate weights into the distance function of each site. 
Three of the most common are additive, power, and multiplicative Voronoi diagrams (see \cite{Franz2013}).
(For brevity, throughout we use the prefix \emph{multiplicative}, rather than the more common \emph{multiplicatively-weighted}.)
For additive Voronoi diagrams the distance from a point $x$ in the plane to a given site $p$ is $d(x,p) = ||x-p||+\alpha_p$, 
for some constant $\alpha_p$, which can vary for each site.  
For power diagrams the distance is given by $d(x,p) = ||x-p||^2-\alpha_p^2$. 
For multiplicative diagrams the distance is given by $d(x,p) = \alpha_p\cdot ||x-p||$.
The worst-case complexity for additive and power diagrams is only linear. 
Here our focus is on multiplicative diagrams, whose worst-case complexity is known to be $\Theta(n^2)$~\cite{AurenhammerE84}. 
These diagrams are used to model, for example, crystal growth where crystals grow together from a set of sites at different rates.

\myparagraph{Previous expected complexity bounds.}
There are many previous results on the expected complexity of various geometric structures under one type of randomness assumption or another.  
Here we focus on previous results relating to Voronoi diagrams.  
For point sites in $\Re^d$, the worst-case complexity of the Voronoi diagram is known to be $\Theta(n^{\lceil d/2\rceil})$, 
however, if the sites are sampled uniformly at random from a $d$-ball or hypercube, for constant $d$, then the expected complexity is $O(n)$ \cite{bdms-accvd-05,d-hdvdl-89}.
For a set of $n$ point sites on a terrain consisting of $m$ triangles, Aronov \textit{et al.}~\cite{AronovBT08} showed that, if the terrain satisfies certain realistic input assumptions then the worst-case complexity of the geodesic 
Voronoi diagram is $\Theta(m+n\sqrt{m})$.
Under a relaxed set of assumption, Driemel \textit{et al.}~\cite{DriemelHR16} showed that the expected complexity is only $O(n+m)$, when the sites are sampled uniformly at random from the terrain domain.

Relevant to the current paper, Har-Peled and Raichel~\cite{Har-PeledR15}  
showed that, for any set of point site locations in the plane and any set of weights, if the weights are assigned to the points according to a random permutation, 
then the expected complexity of the multiplicative Voronoi diagram is $O(n\log^2 n)$. 
The motivation for this work came from Agarwal \textit{et al.}\cite{AgarwalHKS14}, 
who showed that if one randomly fattens a set of segments in the plane, by taking the Minkowski sum of each segment with a ball of random radius, then the expected complexity of the union is near linear, despite having quadratic worst-case complexity.
In a follow-up work to \cite{Har-PeledR15}, Chang \textit{et al.}\cite{ChangHR16} defined the candidate diagram, a generalization of weighted diagrams to multi-criteria objective functions, 
and showed that under similar randomized input assumptions the expected complexity of such diagrams is $O(n\operatorname{polylog} n)$.

\myparagraph{Our results and significance.}
Our first result concerns semi Voronoi diagrams, where the visible region of each site is a half-plane whose bounding line passes through the site, and where the worst-case complexity is quadratic.  
For any set of site locations and bounding lines, we show that if the visible side of each site's bounding line is sampled uniformly at random, 
then the expected complexity of the semi Voronoi diagram is linear, and $O(n\log^3 n)$ with high probability. 
To achieve this we argue that our randomness assumption implies that any point in the plane is likely to be visible by one of its $k$ nearest neighbors, for large enough $k$.
Thus within each cell of the order-$k$ Voronoi diagram, 
one can argue the complexity of semi Voronoi diagram is $O(k^2)$, and so summing over all $O(nk)$ cells gives an $O(nk^3)$ bound.
This is a variant of the candidate diagram approach introduced in \cite{Har-PeledR15}. 
Unfortunately, in general this approach requires $k$ to be more than a constant, which will not yield the desired linear bound. 
Thus here, in order to get a linear bound, we give a new refined version of this approach, by carefully allowing $k$ to vary as needed, 
in the end producing a diagram which is the union of order-$k$ cells, for various values of $k$.

The main focus of the paper is the second part  
concerning multiplicative Voronoi diagrams, where the distance to each site is the Euclidean distance multiplied by some site-dependent weight, and where the worst-case complexity is quadratic. 
We first argue that if the weights are sampled from a set with constant size, then interestingly a similar refined candidate diagram approach yields a linear bound. 
Our main result, however, considers the more challenging case  
when no restrictions are made on the weights, but instead the site locations are sampled uniformly at random from the unit square.
For this case we show that in any grid cell of side length $1/\sqrt{n}$ in the unit square the expected complexity of the multiplicative diagram is $O(1)$.
This implies an $O(n)$ expected bound for the multiplicative diagram over the whole unit square. 
This improves over the $O(n\log^2 n)$ bound of \cite{Har-PeledR15} for this case, by using an interesting new approach which ``stretches'' sites about a given cell based on their weights, 
thus approximately transforming the weighted diagram into an unweighted one with respect to a given cell. As this main result is technically very challenging, it is placed last.

It is important to note the significance of our bounds being linear. 
Given a randomness assumption, one wishes to prove an optimal expected complexity bound, and for the diagrams we consider linear bounds are immediately optimal. 
However, even when a linear bound seems natural, many standard approaches to bounding the expected complexity yield additional $\log$ factors.
That is, not only are linear bounds desirable, but also they are nontrivial to obtain. However, when they are obtained, they reveal the true expected complexity of the structure, free from artifacts of the analysis.

\section{ Preliminaries}\seclab{prelim}
\myparagraph{The standard Voronoi diagram:}
Let $S=\{s_1,s_2,...,s_n \}\subset \Re^2$ be a set of $n$ point \emph{sites} in the plane.
Let $\dist{x}{y}$ denote the Euclidean distance from $x$ to $y$, 
and for two sites $s_i,s_j\in S$ let $\beta(s_i,s_j)$ denote their bisector, that is the set of points $x$ in plane such that $\dist{s_i}{x}=\dist{s_j}{x}$. 
Any site $s_i\in S$ induces a distance function $f_i(x) = \dist{s_i}{x}$ defined for any point $x$ in the plane.
For any subset $T \subseteq S$, the \emph{Voronoi cell} of $s_i\in T$ with respect to $T$, $\VorCell{s_i}{T}=\{x\in \Re^2 \mid \forall s_j \in T \quad f_{i}(x) \leq f_{j}(x) \}$, 
is the locus of points in the plane having $s_i$ as their closest site from $T$. 
We define the \emph{Voronoi diagram} of $T$, denoted $\VorX{T}$, as the partition of the plane into Voronoi
cells induced by the minimization diagram (see \cite{go-hdcg-04}) of the set of distance functions $\{f_i \mid s_i \in T\}$,
that is the projection onto the plane of the lower envelope of the surfaces defined by these bivariate functions.

One can view the union, $U$, of the boundaries of the cells in the Voronoi diagram as a planar graph.
Specifically, define a \emph{Voronoi vertex} as any point in $U$ which is equidistant to three sites in
$S$ (happening at the intersection of bisectors). For simplicity, we make the general position assumption
that no point is equidistant to four or more sites. Furthermore, define a
\emph{Voronoi edge} as any maximal connected subset of $U$ which does not contain a Voronoi vertex. 
(For each edge to have two endpoints we include the “point” at infinity, 
i.e., the graph is defined on the stereographic projection of the plane onto the sphere.)
The complexity of the Voronoi diagram is then defined as the total number of 
Voronoi edges, vertices, and cells.  
As the cells 
are simply connected sets, 
which are faces of a straight-line planar graph, the overall complexity is $\Theta(n)$.

\myparagraph{Order-$k$ Voronoi diagram:} 
Let $\SiteSet$ be a set of $n$ point sites in the plane.  
The \emph{order-$k$ Voronoi diagram} of $\SiteSet$ is the 
partition of the plane into cells, where each cell is the locus of
points having the same set of $k$ nearest sites of $\SiteSet$ 
(the ordering of these $k$ sites by distance can vary within the cell). 
It is not hard to see that this again defines a straight-line partition of the plane into 
cells where the edges on the boundary of a cell are composed of bisector pieces.
The worst-case complexity of this diagram is $\Theta\pth{k(n-k)}$ 
(see \cite[Section 6.5]{Franz2013}).

We also consider the diagram where every point in a cell not only has the same $k$ nearest sites, but also the same ordering of distances to these sites.
We refer to this as the order-$k$ \emph{sequence} Voronoi diagram, known to have $O(nk^3)$ worst-case complexity (see \apndref{shallow}).

\myparagraph{Semi Voronoi diagram:}
Let $S=\{s_1,s_2,...,s_n \}$  be a set of $n$ point sites in the plane, 
where for each $s_i$ there is an associated closed half-plane $\half(s_i)$, whose boundary passes through $s_i$. 
We use $L(s_i)$ to denote the bounding line of $\half(s_i)$.
For any point $x$ in the plane and any site $s_i\in S$, we say that $x$ and $s_i$ are visible to each other when $x\in \half(s_i)$.
Given a point $x$, let $S(x)=\{s_i\in S \mid x\in \half(s_i)\}$ denote the set of sites which are visible to $x$.

For any subset $T \subseteq S$, define the \emph{semi Voronoi cell} of $s_i\in T$ with respect to $T$ as, $\SVorCell{s_i}{T}=\{x\in \half(s_i) \mid \forall s_j \in T\cap S(x) \quad \dist{x}{s_i} \leq \dist{x}{s_j} \}$.
It is possible that there are points in the plane which are not visible by any site in $S$. 
If desired, this technicality can be avoided by adding a pair of sites far away from $S$ which combined can see the entire plane. 
As before, semi Voronoi cells define a straight-line partition of the plane, where now the 
edges on the boundary of a cell are either portions of a bisector or of a half-plane boundary.
In the worst case, the semi Voronoi diagram can have quadratic complexity~\cite{FanLWZ14}.

\myparagraph{Random semi Voronoi diagram:}
We consider semi Voronoi diagrams where the set of sites $S=\{s_1,\ldots,s_n\}$ is allowed to be any fixed set of $n$ points in general position.
For each site $s_i$, the line bounding the half-plane of $s_i$, $L(s_i)$, is allowed to be any fixed line in $\Re^2$ passing through $s_i$.  
Such a line defines two possible visible closed half-spaces.  
We assume that independently for each site $s_i$, one of these two spaces is sampled uniformly at random. 

An alternative natural assumption is that the normal of the half-plane for each site is sampled uniformly at random from $[0,2\pi)$. 
Note our model is strictly stronger, that is any bound we prove will imply the same bound for this alternative formulation.  
This is because one can think of sampling normals from $[0,2\pi)$, as instead first sampling directions for the bounding lines from $[0,\pi)$, and then 
sampling one of the two sides of each line for the normal.

\myparagraph{Multiplicative Voronoi diagram:} 
Let $S=\{s_1,s_2,...,s_n \}$ be a set of $n$ point sites in the plane, 
where for each site $s_i$ there is an associated weight $w_i>0$.
%
Any site $s_i\in S$ induces a distance function $f_i(x) = w_i\cdot \dist{s_i}{x}$ defined for any point $x$ in the plane.
For any subset $T \subseteq S$, the \emph{Voronoi cell} of $s_i\in T$ with respect to $T$, $\VorCell{s_i}{T}=\{x\in \Re^2 \mid \forall s_j \in T \quad f_{i}(x) \leq f_{j}(x) \}$, 
is the locus of points in the plane having $s_i$ as their closest site from $T$. 
The \emph{multiplicative Voronoi diagram} of $T$, denoted $\WVorX{T}$, is the partition of the plane into Voronoi
cells induced by the minimization diagram of the distance functions $\{f_i \mid s_i \in T\}$.

Note that unlike the standard Voronoi diagram, the bisector of two sites is in general an Apollonius disk, potentially leading to disconnected cells.  
Ultimately, the diagram still defines a planar arrangement and so its complexity, denoted by $|\WVorX{S}|$, can still be defined as the number of edges, 
faces, and vertices of this arrangement. 
Note the edges are circular arcs and straight line segments and thus are still constant complexity curves.
In the worst case, the multiplicative Voronoi diagram has $\Theta(n^2)$ complexity \cite{AurenhammerE84}.

\section{The Expected Complexity of Random Semi Voronoi Diagrams}\seclab{semi}
\subsection{The probability of covering the plane}

As it is used in our later calculations, we first bound the probability that for a given subset $X$ of $k$ sites of $S$ that there exists a point in the plane not visible to any site in $X$. 

\begin{lemma}\lemlab{probability}
For any set $X=\{x_1,\ldots,x_k\}$ of $k$ sites $\Prob{(\bigcup_{x_i\in X} \SVorCell{x_i}{X} )\neq \Re^2} \leq (k(k+1)+2)/2^{k+1}= O(k^2/2^k)$.
\end{lemma}
\begin{proof}
Consider the arrangement of the $k$ bounding lines $L(x_1),\ldots,L(x_k)$.
Let $\mathcal{F}$ denote the set of faces in this arrangement (i.e., the connected components of the complement of the union of lines), 
and note that $|\mathcal{F}| \leq k(k+1)/2+1 = O(k^2)$.
Observe that for any face $f\in \mathcal{F}$ and any fixed site $x_i\in X$, either every point in $f$ is visible by $x_i$ or no point in $f$ is visible by  $x_i$. 
Moreover, the probability that face $f$ is not visible by site $x_i$ is $\Prob{\half(x_i) \cap f =\emptyset} \leq 1/2$. 
Hence the probability that a face $f$ is not visible by any of the $k$ sites is
\[
\Prob{\bigcup_{x_i\in X} \half(x_i) \cap f =\emptyset}  \leq 1/2^k.
\]
Hence the probability that at least one face in $\mathcal{F}$ is is not visible by any sites in $X$ is
\[
\Prob{\!\pth{\bigcup_{x_i\in X} \SVorCell{x_i}{X} } \neq \Re^2} \leq \sum_{f\in \mathcal{F}} \Prob{\bigcup_{x_i\in X} \half(x_i) \cap f =\emptyset} 
\leq \!\frac{k(k\!+\!1)\!+\!2}{2^{k+1}}
= O\pth{\frac{k^2}{2^k}}.
\]
\end{proof}

\subsection{A simple near linear bound}

Ultimately we can show that the expected complexity of a random semi Voronoi diagram is linear, however, 
here we first show that \lemref{probability} implies a simple near linear bound which also holds with high probability.
Specifically, we say that a quantity is bounded by $O(f(n))$ with high probability, if for any constant $\alpha$ there exists a constant $\beta$, 
depending on $\alpha$, such that the quantity is at most $\beta \cdot f(n)$ with probability at least $1-1/n^\alpha$.

\begin{lemma} \lemlab{polybound}
 Let $S=\{s_1,\ldots, s_n\}\subset \Re^2$ be a set of $n$ sites, where each site has a corresponding line $L(s_i)$ passing through $s_i$. 
 For each $s_i$, sample a half-plane $\half(s_i)$ uniformly at random from the two half-planes whose boundary is $L(s_i)$.
 Then the expected complexity of the semi Voronoi diagram on $S$ is $O(n \log^3 n)$, 
 and moreover this bound holds with high probability.
\end{lemma}

\begin{proof}
Let $k=c \log n$, for some constant $c$. 
Consider the order-$k$ Voronoi diagram of $S$. First triangulate this diagram so the boundary of each cell has constant complexity. 
(Note triangulating does not asymptotically change the number of cells.) Fix any cell $\Delta$ in this triangulation, 
which in turn fixes some (unordered) set $X$ of $k$-nearest sites. 
By \lemref{probability},
\[
 \Prob{\pth{\bigcup_{x_i\in X} \SVorCell{x_i}{X}} \neq \Re^2} = O(k^2/2^k) = O((c\log n)^2/2^{c\log n}) =O(1/n^{c-\eps'}),
\]
for any arbitrarily small value $\eps'>0$.
Thus with polynomially high probability for every point in $\Delta$ its closest visible site will be in $X$.
Now let $T$ be the set of all $O(k(n-k)) = O(n\log n)$ triangles in the triangulation of the order-$k$ diagram. 
Observe that the above high probability bound applies to any triangle $\Delta\in T$. 
Thus taking the union bound we have that with probability at least $1-1/n^{c-(1+\eps)}$ (where $\eps>\eps'>0$ is an arbitrarily small value), 
simultaneously for every triangle $\Delta$, every point in $\Delta$ will be visible by one of its $k$ closest sites. 
Let $e$ denote this event happening (and $\overline{e}$ denote it not happening).
Conditioning on $e$ happening, there are only $O(k)=O(\log n)$ relevant sites which contribute to the semi Voronoi diagram of any cell $\Delta$. 
Thus the total complexity of the semi Voronoi diagram restricted to any cell $\Delta$ is at most $O(\log^2 n)$, 
as the semi Voronoi diagram has worst-case quadratic complexity~\cite{FanLWZ14}. (Note that as $\Delta$ is a triangle,  
we can ignore the added complexity due to clipping the semi Voronoi diagram of these sites to $\Delta$.)
On the other hand, if $\overline{e}$ happens, then 
the worst-case complexity of the entire semi Voronoi diagram is still $O(n^2)$.

Now the complexity of the semi Voronoi diagram is bounded by the sum over the cells in the triangulation of the complexity of the diagram restricted to each cell.
Thus the above already implies that with high probability the complexity of the semi Voronoi diagram is $O\pth{\sum_{\Delta\in T} \log^2 n} = O(n\log^3 n)$.
As for the expected value, by choosing $c$ sufficiently large,
\begin{align*}
  &= \Ex{|\SVorX{S}| \sep e}\Prob{e} + \Ex{|\SVorX{S}| \sep {\overline{e}} }\Prob{\overline{e}}
 = O\pth{\sum_{\Delta\in T} \log^2 n}\Prob{e} + O(n^2)\Prob{\overline{e}}\\
 &= O(n\log^3 n)\Prob{e} + O(n^2)\Prob{\overline{e}}
 = O(n\log^3 n)\Prob{e} + O(n^2)\cdot (1/n^{c-(1+\eps)})\\
 &= O(n\log^3 n) + O(1/n^{c-(3+\eps)}) = O(n\log^3 n).
\end{align*}
\end{proof}

\subsection{An optimal linear bound}

The previous subsection partitioned the plane based on the order-$k$ Voronoi diagram, for $k=c\log n$, 
and then argued that simultaneously for all cells one of the $k$ nearest sites will be visible. 
This argument is rather coarse, but instead of using a fixed large value for $k$, if one is more careful 
and allows $k$ to vary, then one can argue the expected complexity is linear.
Note that in the following, rather than using the standard unordered order-$k$ Voronoi diagram, we use the more refined 
order-$k$ \emph{sequence} Voronoi diagram as defined above.

\begin{theorem}\thmlab{refined}
 Let $S=\{s_1,\ldots, s_n\}\subset \Re^2$ be a set of $n$ sites, where each site has a corresponding line $L(s_i)$ passing through $s_i$. 
 For each $s_i$, sample a half-plane $\half(s_i)$ uniformly at random from the set of two half-planes whose boundary is $L(s_i)$ (i.e., each has $1/2$ probability).
 Then the expected complexity of the semi Voronoi diagram on $S$ is $\Theta(n)$. \thmlab{linear}
\end{theorem}
\begin{proof}
Consider the partition of the plane by the first order Voronoi diagram of $S$, i.e., the standard Voronoi diagram. We iteratively refine the cells of this partition into higher order sequence Voronoi diagram cells. At each iteration we have a collection of order-$i$ sequence Voronoi diagram cells, and for each cell we either mark it final and stop processing, or further refine the cell into its constituent order-$(i+1)$ sequence cells. Specifically, a cell $\Delta$ is marked final in the $i$th iteration if every point in $\Delta$ is visible by one of its $i$ nearest sites, or when $i=n$ and the cell cannot be refined further.
The process stops when all cells are marked final. 
Below we use $F^k$ to denote the set of cells which were marked final in the $k$th iteration.

Let $C^k$ be the set of all cells of the order-$k$ sequence Voronoi diagram of $S$.
Note that the order-$k$ sequence Voronoi diagram can be constructed by iteratively refining lower order cells, and hence any cell seen at any point in the above process is a cell of the order-$k$ sequence diagram for some value of $k$. 
So consider any cell $\Delta_j\in C^k$ of the order-$k$ sequence diagram of $S$. Let $\#(\Delta_j)$ be the number of order-$(k+1)$ sequence diagram cells inside $\Delta_j$, 
and let $X_j$ be the indicator variable for the event that $\Delta_j$ is refined into its constituent order-$(k+1)$ sequence cells in the $k$th round of the above iterative process.
Note that $\Prob{X_j=1}$ is upper bounded by the probability that $\Delta_j$ has not been marked final by the end of the $k$th round, 
which in turn is bounded by \lemref{probability}. Thus letting $Z^{k+1}$ be the random variable denoting the total number of order-$(k+1)$ sequence cells created in the above process, we have
\[
 \Ex{Z^{k+1}}=\Ex{\sum_{\Delta_j\in C^{k}} \!\! \#(\Delta_j)\cdot X_j} =\!\! \sum_{\Delta_j\in C^{k}}\!\! \#(\Delta_j) \cdot \Ex{X_j}  
 = O\!\left(\frac{k^2}{2^k}\right) \cdot \! \sum_{\Delta_j\in C^{k}}\!\!\! \#(\Delta_j)= O\!\left(\frac{nk^5}{2^k}\right),
\]
as $\sum \#(\Delta_j)$ is at most the total number of cells in the order-$k$ sequence Voronoi diagram, which is bounded by $O(nk^3)$ (see \apndref{shallow}).

Using the same argument as in \lemref{polybound}, the complexity of the random semi Voronoi diagram is bounded by $\sum_{k=1}^n |F^k|\cdot O(k^2)$.
(If $k=n$ then a cell may not be fully visible, though the bound still applies as the worst-case complexity is $O(n^2)$.)
Thus by the above, the expected complexity is 
\[
 \Ex{|\SVorX{S}|} \leq \sum_{k=1}^n \Ex{|F^k|}\cdot O(k^2)  \leq \sum_{k=1}^n \Ex{Z^k}\cdot O(k^2)  \leq O\left(\sum_{k=1}^n \frac{nk^7}{2^k}\right)=O(n). 
\]
\end{proof}

\begin{remark}
 For simplicity the results above were presented for semi Voronoi diagrams, though they extend to the more general VRVD case, where the visibility region of $s_i$ is determined by a cone with base point $s_i$ and angle $\alpha_i$, where the orientation of the cone is sampled uniformly at random. Specifically, the plane can be covered by a set of $(2\pi)/(\alpha_i/2) = 4\pi / \alpha_i$ cones around $s_i$ each  with angle $\alpha_i/2$. Any one of these smaller cones is completely contained in the randomly selected $\alpha_i$ cone with probability $\alpha_i / (4\pi)$. If the $\alpha_i$ are lower bounded by a constant $\beta$, one can then prove a variant of \lemref{probability} (and hence \lemref{polybound} and \thmref{refined}), as the arrangement of all these smaller cones still has $O(k^2)$ complexity, and the probability a face is not visible to any site is still exponential in $k$ but with base proportional to $1-\beta/(4\pi)$.
\end{remark}

\section{The Expected Complexity of Multiplicative Voronoi Diagrams}\seclab{weighted}

In this section we consider the expected complexity of multiplicative Voronoi diagrams under different randomness assumptions.  
First, by using the approach from the previous section, we show that for any set of site locations, if each site samples its weight from a set of constant size $c$, then the expected complexity is $O(n c^6)$.
Next, we consider the case when the sites can have arbitrary weights, but the site locations are sampled uniformly at random from the unit square.
As making no assumptions on the weights makes the problem considerably more challenging, as a warm-up, we first assume the weights are in an interval $[1,c]$.
In this case, we consider a $1/\sqrt{n}$ side length grid, and argue that locally in each grid cell the expected complexity of the multiplicative diagram is $O(c^4)$ (inspired by the approach in \cite{DriemelHR16}), 
and thus over the entire unit square the complexity is $O(n c^4)$. 
Finally, we remove the bounded weight assumption, and argue that for any set of weights, the expected complexity is linear when site locations are 
sampled uniformly at random from the unit square, by introducing the notion of ``stretched'' sites.

\subsection{Sampling from a small set of weights}

\begin{lemma}\lemlab{finite}
Let $W=\{w_1,w_2,...,w_c\}$ be a set of  non-negative real weights and let $S=\{s_1,s_2,s_3,\allowbreak \ldots \allowbreak,s_n\}$ be a set of point sites in the plane, where each site in $S$ is assigned a weight independently and uniformly at random from $W$.
Then the expected complexity of the multiplicative Voronoi diagram of $S$ is $O(n \cdot c^6)$. 
\end{lemma}
\begin{proof}
Consider the partition of the plane determined by the unweighted first order Voronoi diagram of $S$, i.e., the standard Voronoi diagram. 
Following the strategy of the proof of \thmref{linear}, we iteratively refine the cells of this partition into higher order sequence Voronoi diagram cells, 
except now a cell $\Delta$ is marked final in the $i$th iteration if at least one of its $i$-nearest sites has weight $w_m =\min\{w_1,w_2,w_3,...,w_c \}$.
Note that the probability that $k$ sites are all assigned weight larger than $w_m$ is $(1-1/c)^k$.

Using the same notation as in the proof of \thmref{linear}, we have
\begin{align*}
\Ex{Z^{k+1}}
&=\Ex{\sum_{\Delta_j\in C^{k}} \!\! \#(\Delta_j)\cdot X_j} 
= \sum_{\Delta_j\in C^{k}}\! \#(\Delta_j) \cdot \Ex{X_j}\\  
&= O\left((1-1/c)^k\right) \cdot \! \sum_{\Delta_j\in C^{k}} \#(\Delta_j)
= O\left({nk^3}{(1-1/c)^k}\right).
\end{align*}
As the worst-case complexity of the multiplicative Voronoi diagram of $k$ sites is $O(k^2)$, 
overall the complexity of the multiplicative Voronoi diagram is bounded by $\sum_{k=1}^n |F^k|\cdot O(k^2)$. 
Thus by the above, the expected complexity is 
\[
 \Ex{|\WVorX{S}|} \leq \sum_{k=1}^n \Ex{|F^k|}\cdot O(k^2)  
 \leq \sum_{k=1}^n \Ex{Z^k}\cdot O(k^2)  
 \leq O\left(\sum_{k=1}^n {nk^5}{(1-1/c)^k}\right)
 =O(n \cdot c^6),
\]
%
where the last step is obtained by viewing the sum as a power series in $x=(1-1/c)$,
\begin{align*}
&\!\!\!\!\!\!\!\!\sum_{k>0} k^5 (1-1/c)^k = 
\sum_{k>0} k^5 x^k 
= \pth{x \cdot \frac{d}{dx}} \sum_{k>0} k^4 x^k
= \pth{x \cdot \frac{d}{dx}}^5 \sum_{k>0} x^k 
= \pth{x \cdot \frac{d}{dx}}^5 \frac{x}{1-x}\\
&\!\!\!\!\!\!\!\!= \frac{(x^5+26x^4+66x^3+26x^2+x)}{(1-x)^6}
=120c^6-360c^5+390c^4-180c^3+31c^2-c = O(c^6).
\end{align*}
\end{proof}

\subsection{Sampling site locations with bounded weights}

In this section we argue that the expected complexity of the multiplicative Voronoi diagram is linear when the site locations are uniformly sampled and the weights are in a constant spread interval.  
In the next section we remove the bounded weight assumption.
Thus the current section can be viewed as a warm-up, and serves to illustrate the extent to which assuming bounded weights simplifies the problem. However, as the results of the next section subsume those here, this section can be safely skipped if desired. 

The following fact is used both in the proof of the lemma below and the next section.

\begin{fact}\factlab{varHelp}
Consider doing $m$ independent experiments, where the probability of success for each experiment is $\alpha$. 
Let $X$ be the total number of times the experiments succeed.  Then $\Ex{X^2} \leq \alpha m + \alpha^2 m^2 = \Ex{X}+\Ex{X}^2$.
\end{fact}

Note that the above fact holds since $X$ is a binomial random variable, $Bin(\alpha,m)$, and so 
$\Ex{X^2} = \text{\bf Var}[X]+\Ex{X}^2 = m\alpha(1-\alpha)+(m\alpha)^2$.

\begin{lemma}
\lemlab{boundedWeights}
Let $S=\{s_1,s_2,s_3,...,s_n\}$ be a set of point sites in the plane, where for some value $c\geq 1$, each site in $S$ is assigned a weight $w_i\in [1,c]$.
Suppose that the location of each site in $S$ is sampled uniformly at random from the unit square $U$. 
Then the expected complexity of the multiplicative Voronoi diagram of $S$ within $U$ is $O(n \cdot c^4)$.
\end{lemma}
\newcommand{\boundedweightsproof}{
\begin{proof}
Place a regular grid over $U$, where grid cells have side length $1/\sqrt{n}$.
Fix any grid cell $\square = \square_{x,y}$, where $(x,y)\in \sqrt{n}\times \sqrt{n}$.  
We now argue that the expected complexity of the multiplicative Voronoi diagram within $\square$ is $O(c^4)$, 
and thus by linearity of expectation, the expected complexity over all $n$ grid cells in $U$ is $O(n c^4)$.

Let $\rho$ be the random variable denoting the unweighted distance of the closest site in $S$ to the center of the grid cell $\square$, 
and let $X_\rho$ be the random variable denoting the number of points which contribute to the multiplicative Voronoi diagram in $\square$ conditioned on the value $\rho$.
(For now ignore the contribution of the point at distance exactly $\rho$.)
Observe that any point in $S$ which contributes to the multiplicative Voronoi diagram within $\square$ must lie within the annulus centered at the center of $\square$, 
with inner radius $\rho$ and outer radius $c(\rho+\sqrt{2/n})$. 
Conditioned on the value $\rho$, let $\alpha_\rho$ be the probability for a point to fall in this annulus, and let $X_\rho$ be the binomial random variable $Bin(\alpha_\rho, n)$ 
representing the number of points which fall into this annulus.
For now assume $\rho\leq 1/4$, in which case the region outside the disk centered at the center of $\square$ and with radius $\rho$ has area at least $3/4$.
We then have that 
\begin{align*}
 \Ex{X_\rho}&\leq n \alpha_{\rho} \leq n\frac{(\pi (c(\rho+\sqrt{2/n}))^2- \pi\rho^2)}{3/4}\\
 &= O(n c^2(\rho^2(1-1/c^2)+\rho/\sqrt{n}+1/n ))
 = O(c^2+n(c\rho)^2).
\end{align*}
Let $Y_\rho$ be the random variable denoting the complexity of the multiplicative Voronoi diagram within $\square$ when conditioned on the value $\rho$. 
As the worst-case complexity of the multiplicative Voronoi diagram is quadratic, we have $Y_\rho = O((1+X_\rho)^2) = O(1+X_\rho+X_\rho^2)$, where the plus 1 counts the point at distance exactly $\rho$.
Thus using \factref{varHelp}, and again assuming $\rho\leq 1/4$, 
\[
 \Ex{Y_\rho} = O(1+\Ex{X_\rho}+\Ex{X_\rho^2}) = O(c^2+n(c\rho)^2 + (c^2+n(c\rho)^2)^2) = O(c^4(1+n\rho^2+n^2\rho^4)).
\]
Now consider the event that $\rho\in [i/\sqrt{n},(i+1)/\sqrt{n}]$ for some integer $i$. 
For this to happen, the open disk with radius $i/\sqrt{n}$ centered at the center of $\square$ must be empty, 
and one of the $n$ points must lie in the annulus with inner radius $i/\sqrt{n}$ and outer radius $(i+1)/\sqrt{n}$.  Therefore,
\begin{align*}
 \Prob{\rho\in [i/\sqrt{n},(i+1)/\sqrt{n}]} \leq n\cdot (\pi ((i+1)/\sqrt{n})^2 - \pi (i/\sqrt{n})^2)\cdot (1-\pi (i/\sqrt{n})^2)^{n-1}\\
 =\pi(2i+1)\cdot (1-\pi i^2/n)^{n-1}
 \leq \pi(2i+1) e^{-\pi i^2 (n-1)/n}
 \leq \pi(2i+1) e^{-i^2}.
\end{align*}
Furthermore, by the above, when $\rho\in [i/\sqrt{n},(i+1)/\sqrt{n}]$ and $\rho\leq 1/4$, we have 
\[
 \Ex{Y_\rho}=O(c^4(1+n\rho^2+n^2\rho^4)) = O(c^4(1+(i+1)^2+(i+1)^4)),
\]
and for $\rho\geq 1/4$ we have the trivial bound $\Ex{Y_\rho}=O(n^2)$.

Finally, let $Y$ be the random variable denoting the complexity of the multiplicative Voronoi diagram within $\square$.
By the law of total expectation we have, 

\begin{align*}
 \Ex{Y} &\leq \left(\sum_{i=0}^{\sqrt{n}/4 - 1}  \! \Prob{\rho\in \left[\frac{i}{\sqrt{n}},\frac{(i\!+\!1)}{\sqrt{n}}\right]} \cdot O(c^4(1+(i\!+\!1)^2+(i\!+\!1)^4))\right) \\
       & ~~~~+ \left(\sum_{i=\sqrt{n}/4}^{\sqrt{n}} \! \Prob{\rho\in \left[\frac{i}{\sqrt{n}},\frac{(i\!+\!1)}{\sqrt{n}}\right]} \cdot O(n^2)\right)\\
       &\leq \left(\sum_{i=0}^{\sqrt{n}/4 - 1}  \! \frac{\pi(2i+1)}{e^{i^2}} \cdot O(c^4(1+(i\!+\!1)^2+(i\!+\!1)^4))\!\right) 
       + \left(\sum_{i=\sqrt{n}/4}^{\sqrt{n}} \! \frac{\pi(2i+1)}{e^{i^2}} \cdot O(n^2)\!\right)\\
 &= O(1)+\! \sum_{i=0}^{\sqrt{n}/4-1} \frac{\pi(2i+1)}{e^{i^2}} \cdot O(c^4(1+(i\!+\!1)^2+(i\!+\!1)^4))\\ 
 &=O(1)+\! \sum_{i=0}^{\sqrt{n}/4} \frac{2i+1}{e^{i^2}} \cdot O(c^4(i\!+\!1)^4)
 = O(c^4)+ O\left(\sum_{i=1}^{\sqrt{n}/4} \frac{c^4 i^5}{e^{i^2}}\right) = O(c^4).
\end{align*}
\end{proof}%
}
\boundedweightsproof

\subsection{Sampling sites locations in general}

In this section we argue that the expected complexity of the multiplicative Voronoi diagram is linear when the site locations are uniformly sampled. 
Here the weights can be any arbitrary set of positive values.
Without loss of generality we can assume the smallest weight is exactly $1$, as dividing all weights by the same positive constant does not change the diagram.
Throughout, $m$ denotes the maximum site weight, hence all weights are in the interval $[1,m]$.

Let $\sigma$ be an arbitrary point in the unit square. The high level idea is that we want to apply a transformation to the sites such that the weighted Voronoi diagram around $\sigma$ can be interpreted as an unweighted Voronoi diagram.  Specifically, for a site $s$ with weight $w$ and distance $d=||s-\sigma||$, let the \emph{stretched site} of $s$ with respect to $\sigma$, denoted by $\Stretch$, be the point at Euclidean distance $w\cdot d$ from $\sigma$ which lies on the ray from $\sigma$ through $s$.  That is, the weighted distance from $s$ to $\sigma$ is the same as the unweighted distance from $\Stretch$ to $\sigma$.  

So let $\sigma$ be an arbitrary fixed point in the unit square, $S=\{s_1,\ldots,s_n\}$ be a set of sites with weights $\{w_1,\ldots,w_n\}\subset [1,m]$ whose positions have been uniformly sampled from the unit square, 
and let $T=\{\Stretch_1,\ldots,\Stretch_n\}$ be the corresponding set of stretched sites. 

Let $\gamma=\sqrt{1/(2n)}$. Our goal is to argue that for any arbitrary choice of $\sigma$ in $U$ the expected complexity of the multiplicative Voronoi diagram in the ball $B(\sigma,\gamma)$ is constant, 
i.e., $\Ex{|\WVorX{S}\cap B(\sigma,\gamma)|}=O(1)$. Then by the grid argument in the previous section this immediately implies a linear bound on the expected complexity of the overall diagram in $U$.  
Namely, place a uniform grid over $U$, where the grid cell side length is $1/\sqrt{n}$.  Then as each cell is contained in a ball $B(\sigma,\gamma)$ for some $\sigma$, and there are $n$ cells overall, 
by the linearity of expectation the expected complexity of the overall diagram is $O(n)$.

At a high level, the idea is simple.  We wish to argue that sites whose stretched location is far from $\sigma$ will be blocked from contributing by sites whose stretched location is closer to $\sigma$.  
However, putting this basic plan into action is tricky and requires handling various edge cases. 
In particular, later it will become clear why we need to consider the following cases for where sites lie relative to $\sigma$.

\begin{definition}\deflab{sets}
For a fixed point $\sigma$ in $U$, let $X$ be the subset of $S$ which falls in $B(\sigma,4\gamma)$, i.e., $X=S\cap B(\sigma,4\gamma)$. Let $Y$ be the complement set, i.e., $Y=S\setminus X$.  
\end{definition}

\begin{remark}\remlab{fixed}
In the following, we typically condition on the set of sites which fall in $Y$ as being fixed, but not the actual precise locations of those sites.
We refer to this as ``Fixing $Y$''. Ultimately the statements below hold regardless of what sites actually fall in $Y$.
\end{remark}

Fix $Y$.  Let $r_1$ be the radius  such that the expected number of stretched sites from $Y$ that are contained in $B(\sigma,r_1)$ is $n\cdot \pi(16\gamma)^2$.
%
Observe that a site can only be moved further from $\sigma$ after stretching, thus $r_1\geq 16\gamma$.
(Potentially $r_1$ is significantly larger.)   
Also, note that as $Y$ is fixed, the value of $r_1$ is fixed.

For any integer $j>0$, let $D_j$ be the disk with radius $r_j=r_1 \times 2^{j-1}$, 
centered at $\sigma$. 
%
Moreover, define the rings $R_{j+1} = D_{j+1} \setminus D_{j}$, for any $j> 0$.

\begin{lemma}
\lemlab{move}
Consider two sites $s_j \in Y$, and $s_i$ such that either $s_i\in Y$ or $s_i\in X$ with $w_i \leq w_j$. 
For any $k'\geq k\geq 1$, if $\Stretch_i\in D_k$ and $\Stretch_j\in R_{k'+2}$, then $s_j$ cannot contribute to the multiplicative Voronoi diagram in $B(\sigma,\gamma)$.
\end{lemma}
\begin{proof}
    For site $s_i$, let $d_i=||s_i-\sigma||$ and $d'_i=w_i\cdot d_i$.  Similarly define $d_j$ and $d_j'$ for site $s_j$.
    Note that the furthest weighted distance of a point in $B(\sigma,\gamma)$ to $s_i$ is $w_i \cdot (d_i+\gamma)$, 
    and the closest weighted distance of a point in $B(\sigma,\gamma)$ to $s_j$ is $w_j \cdot (d_j-\gamma)$.
    Thus it suffices to argue $w_i \cdot (d_i+\gamma) <w_j \cdot (d_j-\gamma)$, or equivalently $(w_i+w_j)\gamma<d'_j-d'_i$.
    To that end, observe $d'_j-d'_i \geq d'_j - d_j'/2= d'_j/2$. Thus we only need to argue $(w_i+w_j)\gamma< d'_j/2$.
    
    Case 1, $w_i \leq w_j$: In this case $(w_i+w_j)\gamma \leq w_j \cdot 2\gamma< w_j\cdot d_j/2 = d'_j/2$.
    
    Case 2, $w_i > w_j$:  
    In this case $(w_i+w_j)\gamma \leq w_i \cdot 2\gamma< w_i \cdot d_i/2 = d'_i/2 \leq d'_j/2 $.
\end{proof}

Note that a site with weight $1$ does not move after stretching. Thus as $S$ always has a site with weight $1$, there must be at least one stretched site in $U$.
Therefore, if we define $Z-1$ to be the smallest value of $j$ such that $U\subseteq D_j$, then $D_{Z-1}$ must contain a stretched site. 
Thus the above lemma implies any site which contributes to the multiplicative Voronoi diagram within $B(\sigma,\gamma)$ must lie within $D_{Z}$ after being stretched, 
and so going forward we only consider stretched sites in $D_{Z}$.

\begin{definition}\deflab{probability}
Fix $Y$, and consider any $1 \leq j\leq Z-2$. For $s_i\in Y$, let $p_{i,j}$ denote the probability that $t_i$ is located in $D_j$.  
For $s_i\in Y$, let $q_{i,j}$ denote the probability that $s_i$ is located in $U\setminus B(\sigma,16\gamma)$ and $t_i$ is located in $R_{j+2}\cup R_{j+1}$, 
conditioned on the event that no stretched sites from $Y$ are located in $D_j$.\footnote{Note for $j \leq Z-2$, $D_j\subsetneq U$, thus the condition that no stretched sites are in $D_j$ has non-zero probability.}
\end{definition}

\begin{lemma}
Fix $Y$. For any $1 \leq j \leq Z-2$ and $s_i\in Y$ we have $q_{i,j} \leq 32\cdot  p_{i,j}$.
\lemlab{relation}
\end{lemma}
\begin{proof}
Note in the following we always take as given that $s_i\in Y$. Observe that 
\begin{align*}
p_{i,j} &= \Prob{t_i\in D_j} = \Prob{s_i\in U\cap B(\sigma,r_j/w_i)}
= \frac{area( (U\cap B(\sigma,r_j/w_i)) \setminus B(\sigma,4\gamma))}{area(U\setminus B(\sigma,4\gamma))}.
\end{align*}

Then as the location of the sites in $Y$ are independent, for $q_{i,j}$ we have,  
\begin{align*}
q_{i,j} &= \Prob{(s_i\in U\setminus B(\sigma,16\gamma)) \cap (t_i\in D_{j+2}\setminus D_{j}) \mid \forall s_k\in Y, t_k\notin D_j}\\
 &= \Prob{(s_i\in U\setminus B(\sigma,16\gamma)) \cap (t_i\in D_{j+2}\setminus D_{j}) \mid t_i\notin D_j}\\
 &= \Prob{(s_i\in U\setminus B(\sigma,16\gamma)) \cap (t_i\in D_{j+2}\setminus D_{j})} / \Prob{t_i\notin D_j}\\
 &= \Prob{s_i\in (U\cap B(\sigma,4 r_j/w_i)) \setminus (B(\sigma, r_j/w_i)\cup B(\sigma,16\gamma))} / \Prob{t_i\notin D_j}\\
 &\leq \Prob{s_i\in (U\cap B(\sigma,4 r_j/w_i)) \setminus B(\sigma,16\gamma)} / \Prob{t_i\notin D_j}\\
 &= \frac{area( (U\cap B(\sigma,4 r_j/w_i)) \setminus B(\sigma,16\gamma)) }{ area(U\setminus B(\sigma,4\gamma)) \cdot \Prob{t_i\notin D_j} }
 \leq \frac{16\cdot area( (U\cap B(\sigma, r_j/w_i)) \setminus B(\sigma,4\gamma)) }{ area(U\setminus B(\sigma,4\gamma)) \cdot \Prob{t_i\notin D_j}}\\
 &= \frac{16\cdot p_{ij}}{ \Prob{t_i\notin D_j} }
 = \frac{16\cdot p_{ij}}{ 1-p_{ij} }.
\end{align*}
If $p_{ij}\leq 1/2$, then the above implies $q_{ij}\leq 16\cdot p_{ij}/ (1-p_{ij})\leq 32\cdot p_{ij}$.
On the other hand, if $p_{ij}> 1/2$, then $p_{ij}> 1/2\geq q_{ij}/2$.
\end{proof}

\begin{fact}\factlab{indep}
Fix $Y.$ For any $1 \leq j\leq Z-2$, consider the event, $C_j$, that $j$  is the largest index such that there are no stretched sites from $Y$ located in $D_j$. 
We have $\Prob{C_j} \leq \prod_{s_i\in Y} (1-p_{i,j})$.
\end{fact}

\begin{lemma}\lemlab{max}
Fix $Y.$ Let $\psi$ be the number of sites that fall outside $B(\sigma,16\gamma)$ and contribute to the multiplicative Voronoi diagram within $B(\sigma,\gamma)$.
Then we have 
\[
\Ex{\psi^2}\leq O(1)+ 2\sum _{j=1}^{Z-2}\left( \left(\prod_{s_i \in Y} (1-p_{i,j})\right) \cdot \left(1+3 \sum_{s_i \in Y} q_{i,j} + \left(\sum_{s_i \in Y} q_{i,j}\right)^2 \right) \right).
\]
\end{lemma}

\begin{proof}
For $1\leq j\leq Z-2$, let $C_j$ be the event that $j$ is the largest index such that $D_j$ contains no stretched sites from $Y$. 
(Note for $i\neq j$, $C_i$ and $C_j$ are disjoint events.)
If $C_j$ occurs then only sites in $R_{j+1} \cup R_{j+2}$ can contribute to the weighted diagram in $B(\sigma,\gamma)$, based on \lemref{move}. 
Let $y_j$ be the random variable denoting the number of sites such that $s_i\notin B(\sigma,16\gamma)$ and $t_i\in R_{j+1}\cup R_{j+2}$.
If $C_j$ occurs, then $y_j$ is an upper bound on the number of sites which fall outside $B(\sigma,16\gamma)$ and contribute to the weighted diagram in $B(\sigma,\gamma)$.  
Note we must also consider the event that $C_j$ does not occur for any $j\geq 1$, namely there exist stretched sites from $Y$ contained in $D_1$.
Call this event $C_0$, and let $y_0$ be the random variable for the number of sites such that $s_i\notin B(\sigma,16\gamma)$ and $t_i\in D_2$.
By the law of total expectation,
\[
\Ex{\psi^2} = \sum _{j=0}^{Z-2} (\Prob{C_j} \cdot \Ex{(y_j)^2\mid C_j}),
\]
where the sum stops at $Z-2$ since there always exists a stretched site with weight $1$ in $D_{Z-1}$ (by definition of $Z$), and so stretched sites outside of $D_Z$ can be ignored.

So consider any $\Ex{(y_j)^2\mid C_j}$ term, for some $j>0$.
Let $A_j$ be the event that $D_j$ contains no stretched sites from $Y$ and let $B_j$ be the event that $R_{j+1}$ has a stretched site from $Y$, and observe that $C_j=A_j\cap B_j$.
The following claim is intuitive, and its proof is in \apndref{claimproof}.
\begin{claim}\claimlab{fix}
 $\Ex{y_j^2 \mid C_j} \leq 2\Ex{(y_j+1)^2\mid A_j}$.
\end{claim}
The events, over all $s_i\in Y$, that $s_i$ falls outside $B(\sigma,16\gamma)$ while $t_i$ is located in $R_{j+1} \cup R_{j+2}$, are independent.  
Moreover, this event for any $s_i\in Y$, when conditioned on $A_j$, has probability $q_{i,j}$ (see \defref{probability}). 
Hence conditioned on $A_j$, the random variable $y_j$ has a Poisson Binomial distribution. 
Thus we have $\Ex{y_{j}\mid A_j}= \sum_{s_i \in Y} (q_{i,j})$, and the variance is then 
$\text{\bf Var}[y_j\!\mid\! A_j] =\sum_{s_i \in Y} \left( q_{i,j}(1-q_{i,j}) \right) \leq \sum_{s_i \in Y} (q_{i,j})=\Ex{y_{j}\!\mid\! A_j}$. Thus using the above claim, 
\begin{align*}
\Ex{(y_{j})^2 \mid C_j} &\leq 2\Ex{(y_j+1)^2\mid A_j} = 2\Ex{1+2y_j+y_j^2 \mid A_j} \\
&= 2(1+2\Ex{y_j\mid A_j}+(\text{\bf Var}[y_j\mid A_j]+\left(\Ex{y_{j}\mid A_j}\right)^2))\\
&\leq 2(1+3\Ex{y_j\mid A_j}+\left(\Ex{y_{j}\mid A_j}\right)^2)
\leq  2\left(1+ 3\sum_{s_i \in Y} q_{i,j} +\left(\sum_{s_i \in Y} q_{i,j}\right)^2\right).
\end{align*}

Suppose that $\Prob{C_0}\cdot \Ex{(y_0)^2|C_0}=O(1)$, then the lemma statement follows.
Specifically, by \factref{indep}, $\Prob{C_j}\leq \prod_{s_i \in Y} (1-p_{i,j})$, for any $1\leq j\leq Z-2$. 
Thus by the above, 
\begin{align*}
\Ex{\psi^2} &= \sum _{j=0}^{Z-2} (\Prob{C_j} \cdot \Ex{(y_j)^2\mid C_j})\\
&\leq O(1)+2\sum _{j=1}^{Z-2} \left( \left( \prod_{s_i \in Y} (1-p_{i,j}) \right) \cdot  \left(1+3 \sum_{s_i \in Y} q_{i,j} +\left(\sum_{s_i \in Y} q_{i,j}\right)^2\right) \right).
\end{align*}
Thus what remains is to show $\Prob{C_0} \cdot \Ex{(y_0)^2|C_0}=O(1)$.
First, observe that 
\begin{align*}
\!\!\!\!\!\!\!\Prob{C_0} \Ex{(y_0)^2|C_0}\!&=\! \Prob{C_0} \sum_{k=1}^{|Y|^2}\! k\cdot \Prob{(y_0)^2=k|C_0}\!=\! \Prob{C_0} \sum_{k=1}^{|Y|^2} \frac{k\cdot \Prob{((y_0)^2=k) \cap C_0}}{\Prob{C_0}}\\ 
&= \sum_{k=1}^{|Y|^2} k\cdot \Prob{((y_0)^2=k) \cap C_0}
\leq \sum_{k=1}^{|Y|^2} k\cdot \Prob{(y_0)^2=k} = \Ex{(y_0)^2}.
\end{align*}%
Thus it suffices to argue $\Ex{(y_0)^2}=O(1)$.

For any $s_i\in Y$, let $q_{i,0}$ denote the probability that $s_i\notin B(\sigma,16\gamma)$ while $t_i\in D_2$, then $\Ex{y_0}=\sum_{s_i\in Y} q_{i,0}$.
Note that as $area( B(\sigma,r_2/w_i) \setminus B(\sigma,16\gamma))  \leq 4 \cdot area( B(\sigma, r_1/w_i) \setminus B(\sigma,4\gamma))$ 
it follows that $area( (U\cap B(\sigma,r_2/w_i)) \setminus B(\sigma,16\gamma))  \leq 4 \cdot area( (U\cap B(\sigma, r_1/w_i)) \setminus B(\sigma,4\gamma))$.
This implies $q_{i,0} \leq 4\cdot p_{i,1}$, and hence $\sum_{s_i\in Y} q_{i,0} \leq 4 \sum_{s_i\in Y} p_{i,1}=4\cdot n\pi(16\gamma)^2=O(1)$, by the definition of $r_1$.
Thus $ \Ex{y_0} =O(1)$, and since $y_0$ has a Poisson Binomial distribution, $ \Ex{(y_0)^2} = (\Ex{y_0})^2+\text{\bf Var}[y_0] \leq (\Ex{y_0})^2+\Ex{y_0}=O(1)$.
\end{proof}
%
%
\begin{lemma}
\lemlab{step}
Let $P_j=\sum_{s_i\in Y} p_{i,j}$. Then for any $j\leq Z-3$, we have $P_j\geq P_{j-1}+1$.
\end{lemma}
\newcommand{\quarterproof}{
\begin{proof}
Note that by definition, the expected number of stretched sites from $Y$ that are contained in $B(\sigma,r_1)$ is $n\pi(16\gamma)^2>2$, and thus $P_1=(\sum_{s_i\in Y} p_{i,1})> 2$.

\begin{figure}[t]
  \centering
    \includegraphics[width=.34\linewidth]{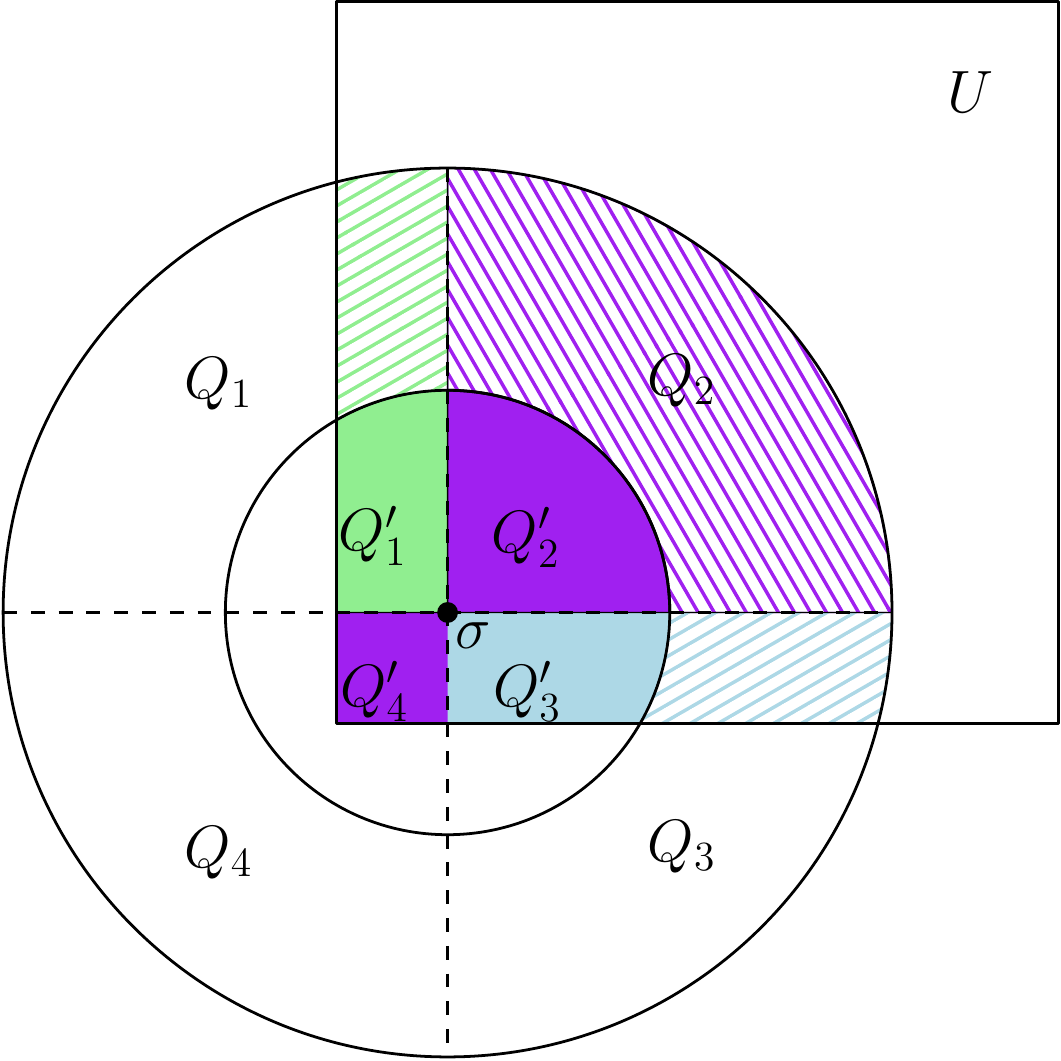}
    \caption{The locations of the four defined quarters for each disk, and how they intersect $U$.}
    \figlab{cut}
\end{figure}

We argue that for any $j\leq Z-3$, that $area(U\cap D_j) \geq 2\cdot area(U\cap D_{j-1})$.
To this end, break $D_j$ into 4 quarters, by cutting it with a vertical and horizontal line through $\sigma$.
Let the quarters be labeled $Q_1$, $Q_2$, $Q_3$, and $Q_4$, in clockwise order, starting with the northwest quarter, see \figref{cut}. 
Similarly break $D_{j-1}$ into quarters, $Q'_1$, $Q'_2$, $Q'_3$, and $Q'_4$, with the same clockwise labeling order.
Recall that by definition $Z-1$ is the smallest $j$ such that $U\subseteq D_j$.
Since $\sigma\in U$, this implies that for any $j\leq Z-3$
at least one of the quarters of $D_j$ is fully contained in $U$, and without loss of generality assume it is $Q_2$.
Note that $area(D_j)=4 area(D_{j-1})$, thus this implies that $area(U\cap Q_2) = area(Q_2) = 2(area(Q'_2)+area(Q'_4))\geq 2(area(U\cap Q'_2)+area(U\cap Q'_4))$.
It is easy to argue that since $r_j=2r_{j-1}$ that $area(U\cap Q_1)\geq 2 area(U\cap Q'_1)$ and $area(U\cap Q_3)\geq 2 area(U\cap Q'_3)$, 
thus summing over all quarters $area(U\cap D_j) \geq 2\cdot area(U\cap D_{j-1})$.

This implies that $area(U\cap B(\sigma,r_j/w_i)) \geq 2\cdot area(U\cap B(\sigma,r_{j-1}/w_i))$, 
which in turn implies $area((U\cap B(\sigma,r_j/w_i)) \setminus B(\sigma,4\gamma)) \geq 2 \cdot area( (U\cap B(\sigma,r_{j-1}/w_i)) \setminus B(\sigma,4\gamma))$.
Hence by \defref{probability}, $p_{i,j} \geq 2p_{i,j-1}$, and therefore $P_j \geq 2 P_{j-1}\geq P_{j-1}+1 $.
\end{proof}
}
\quarterproof

\begin{lemma}\lemlab{constant}
\[
 \sum _{j=1}^{Z-2} \left( \left( \prod_{s_i \in Y} (1-p_{i,j}) \right) \cdot  \left(1+3 \sum_{s_i \in Y} q_{i,j} +\left(\sum_{s_i \in Y} q_{i,j}\right)^2\right) \right) = O(1).
\]
\end{lemma}
\begin{proof}
We argue that $\sum _{j=1}^{Z-2} \left( \left( \prod_{s_i \in Y} (1-p_{i,j}) \right) \cdot \left(\sum_{s_i \in Y} q_{i,j}\right)^\alpha \right)= O(1)$, where $\alpha =0$, $1$, or $2$, thus implying the lemma statement. 
By \lemref{relation}, $q_{i,j}\leq 32\cdot p_{i,j}$, and therefore
\begin{align*}
\sum _{j=1}^{Z-2}\left( \left(\prod_{s_i \in Y} (1-p_{i,j})\right) \cdot \left(\sum_{s_i \in Y} q_{i,j}\right)^\alpha \right)  
\leq \sum _{j=1}^{Z-2} \left(\left(\prod_{s_i \in Y} (1-p_{i,j})\right) \cdot \left(32 \sum_{s_i \in Y} p_{i,j}\right)^\alpha \right) \\
\leq 32^\alpha \sum _{j=1}^{Z-2} \left(\left(e^{-\sum_{s_i \in Y} p_{i,j}}\right) \cdot \left(\sum_{s_i \in Y} p_{i,j}\right)^\alpha \right)
= 32^\alpha \sum _{j=1}^{Z-2} \left(\left(e^{-P_j}\right) \cdot \left(P_j\right)^\alpha \right),
\end{align*}
where the last inequality follows as $1-x\leq e^{(-x)}$, and the last equality is by the definition of $P_j$ from the \lemref{step} statement.

Note that by definition, the expected number of stretched sites from $Y$ that are contained in $B(\sigma,r_1)$ is $n\pi(16\gamma)^2$, and thus $P_1=(\sum_{s_i\in Y} p_{i,1})> 2$.
Moreover, the function $x^\alpha e^{-x}$ is monotonically decreasing for $x>2$, and always has value less than $1$, for $\alpha = 0,1,\text{ or }2$.
Thus since by \lemref{step}, $P_j\geq P_{j-1}+1$ for $j\leq Z-3$, and since $P_{Z-2}\geq P_{Z-3}$, we have, 
\begin{align*}
\sum _{j=1}^{Z-2} P_j^\alpha \cdot e^{-P_j}
\leq 1+ \sum _{j=1}^{Z-3} P_j^\alpha \cdot e^{-P_j}
\leq 1+\sum _{x=2}^{\infty} x^\alpha \cdot e^{-x} \leq 2+\int_{2}^{\infty} x^\alpha\cdot e^{-x} dx \leq 4,
\end{align*}
for $\alpha =0, 1,\text{ or }2$. Combining the above two equalities thus yields the lemma statement.
\end{proof}

Now that we have the above lemmas for any fixed $Y$, we are finally ready to prove our main lemma (where $Y$ is no longer assumed to be fixed).

\begin{lemma}\lemlab{alltogether}
Let $S$ be a set of $n$ point sites in the plane, with arbitrary positive weights.
Suppose that the location of each site in $S$ is sampled uniformly at random from the unit square $U$. 
Then for any point $\sigma\in U$, $\Ex{|\WVorX{S}\cap B(\sigma,\gamma)|}=O(1)$.
\end{lemma}
\begin{proof}
Let $\hat{\Psi} = S\cap B(\sigma, 16\gamma)$ be the sites which fall in $B(\sigma, 16\gamma)$ and let $\Psi = S\setminus \hat{\Psi}$ be the complement set.
Let $\hat{\psi}$, $\psi$,  be the random variables denoting the number of sites respectively from $\hat{\Psi}$, $\Psi$,  which contribute to the multiplicative diagram in $B(\sigma,\gamma)$.

Recall that the worst-case complexity of the multiplicative diagram is quadratic in the number of sites. Thus it suffices to bound
$\Ex{(\hat{\psi}+\psi)^2} \leq \Ex{2(\hat{\psi}^2+\psi^2)} = 2(\Ex{\hat{\psi}^2}+\Ex{\psi^2}).$
Thus we now show each of the above two expected value terms are constant.  
First, note that \lemref{max} and \lemref{constant} combined imply that $\Ex{\psi^2}=O(1)$
(as those lemmas hold regardless of which sites fall in $Y$).
Thus we only need to bound $\Ex{\hat{\psi}^2}$. Observe that clearly $|\hat{\Psi}|\geq \hat{\psi}$. 
To bound $|\hat{\Psi}|$, observe that $area(B(\sigma, 16\gamma))=O(1/n)$, and thus $\Ex{|\hat{\Psi}|} = O(1)$. 
Moreover, the number of sites which fall into this ball is a binomial random variable.  
Thus by \factref{varHelp}, $\Ex{\hat{\psi}^2} \leq \Ex{|\hat{\Psi}|^2} \leq \Ex{|\hat{\Psi}|}+\Ex{|\hat{\Psi}|}^2 = O(1)$.
\end{proof}

Consider placing a uniform grid with side length $1/\sqrt{n}$ over the unit square $U$.  Then for any grid cell, if we set $\sigma$ to be the center of the grid cell then $B(\sigma,\gamma)$ contains the grid cell, 
and hence the above lemma implies the expected complexity in the grid cell is constant.  Thus using linearity of expectation over all $n$ grid cells implies the following main theorem.

\begin{theorem}
Let $S$ be a set of $n$ point sites in the plane, with arbitrary positive weights.
Suppose the location of each site in $S$ is sampled uniformly at random from the unit square $U$. 
Then the expected complexity of the multiplicative Voronoi diagram of $S$ within $U$ is $O(n)$.
\end{theorem}


\bibliographystyle{plain}
\bibliography{voronoi}%

\myparagraph{Acknowledgements.} 
The authors would like to thank Sariel Har-Peled for useful discussions concerning the case of sampled site locations. 


\newpage
\appendix

\newenvironment{reflemma}[1]{\medskip\parindent 0pt{\bf Lemma \ref{#1}.}\em }{\vspace{1em}}
\newenvironment{reftheorem}[1]{\medskip\parindent 0pt{\bf Theorem \ref{#1}.}\em }{\vspace{1em}}
\newenvironment{refclaim}[1]{\medskip\parindent 0pt{\bf Claim \ref{#1}.}\em }{\vspace{1em}}

\section{Claim Proof}\apndlab{claimproof}
Here we prove the claim from \lemref{max}. Below is the restated claim.\\

{\noindent
\textbf{\claimref{fix}.}
\textit{
$\Ex{y_j^2 \mid C_j} \leq 2\Ex{(y_j+1)^2\mid A_j}$.
}

\begin{proof}
Let $y_j$ and $C_j=A_j\cap B_j$ be as defined in the proof of \lemref{max}.
Throughout $j$ is fixed and so we drop the $j$ subscripts.
Thus we must prove
 $\Ex{y^2 | A \cap B} \leq 2\Ex{(y+1)^2 | A}$.

As the conditioning on $A$ appears in all terms, for simplicity of exposition we write that we want to show $\Ex{y^2 | B} \leq 2\Ex{(y+1)^2}$ where the conditioning on $A$ is implicit.


Note that $y=y'+y''$, where $y'$ is the number of stretched sites falling in $R_{j+1}$ and $y''$ the number falling in $R_{j+2}$. Moreover, $(y'+y'')^2 \leq 2((y')^2+(y'')^2)$.

\begin{lemma}
    \lemlab{aa}%
    $\Ex{(y'')^2 \mid y' \neq 0} \leq \Ex{(y'')^2}$.
\end{lemma}
\begin{proof}
    Let $\alpha = \Ex{(y'')^2 \mid y' = 0}$ and
    $\beta = \Ex{(y'')^2 \mid y' \neq 0}$.  
    It is easy to verify that
    $\alpha = \Ex{(y'')^2 \mid y' = 0} \allowbreak \geq \Ex{(y'')^2}$.        
    Now, observe that
    \begin{align*}
      \mu
      &= \Ex{(y'')^2}%
        =%
        \Ex{(y'')^2 \mid y' = 0}
        \Prob{y' = 0}%
          +
        \Ex{y''^2 \mid y' \neq 0} \Prob{y' \neq 0}
      \\&%
      =%
      \alpha 
      \Prob{y' = 0}%
      + \beta 
      (1 - \Prob{y' = 0}).%
    \end{align*}
    Namely, $\mu$ is a convex combination of $\alpha$ and $\beta$,
    and since $\alpha\geq \mu$, it must be that $\beta\leq \mu$, as claimed.
\end{proof}

As $B=(y'\neq 0)$, by the above lemma and linearity of expectation we have 
\begin{align*}
 \Ex{y^2\mid B} \leq 2\Ex{(y')^2 \mid y'\neq 0}+2\Ex{(y'')^2\mid y'\neq 0}
 \leq 2\Ex{(y')^2 \mid y'\neq 0} + 2\Ex{(y'')^2}.
\end{align*}
Observe that $2\Ex{(y'+1)^2}+2\Ex{(y'')^2}\leq  2\Ex{(y+1)^2}$. Thus if we can prove $\Ex{(y')^2\mid y'\neq 0}\allowbreak \leq \Ex{(y'+1)^2}$, then the above implies $\Ex{y^2\mid B}\leq 2\Ex{(y+1)^2}$ as claimed.


So to prove $\Ex{(y')^2\mid y'\neq 0}\leq \Ex{(y'+1)^2}$, note that $(y'\neq 0)=\cup_i X_i$, where $X_i$ is the event that the $i$th stretched site is in $R_{j+1}$. 
Thus,  
\begin{align*}
 E[(y')^2 \mid y'\neq 0] 
 =& E[(y')^2\mid X_1]\cdot Pr[X_1]+E[(y')^2\mid \overline{X_1} \cap X_2]\cdot Pr[\overline{X_1}\cap X_2]+\ldots\\
 &+E[y'^2 \mid \overline{X_1}\cap\ldots \cap\overline{X_{n-1}}\cap X_n]\cdot Pr[\overline{X_1}\cap\ldots \cap\overline{X_{n-1}}\cap X_n]\\
 &+ E[(y')^2 \mid \overline{X_1}\cap\ldots \cap\overline{X_{n}}]\cdot Pr[\overline{X_1}\cap\ldots \cap\overline{X_{n}}].
\end{align*}
Note the last term above is zero and can be ignored.  Also note that
\[
 Pr[X_1]+Pr[\overline{X_1}\cap X_2]+\ldots +Pr[\overline{X_1}\cap\ldots \cap\overline{X_{n-1}}\cap X_n]+Pr[\overline{X_1}\cap\ldots \cap\overline{X_{n}}]=1.
\]
Thus the claim follows if we can argue that each expectation in the above sum is upper bounded by $E[(y'+1)^2]$.
So consider any term $E[(y')^2 \mid \overline{X_1}\cap\ldots \cap\overline{X_{k-1}}\cap X_k]$.
Let $z$ be the number of sites from $\{p_{k+1}\ldots p_n\}$ falling in $R_{j+1}$. Then since the points were sampled independently we have 
 \begin{align*}
 E[(y')^2 \mid \overline{X_1}\cap\ldots \cap\overline{X_{k-1}}\cap X_k] 
 &\leq E[(z+1)^2\mid \overline{X_1}\cap\ldots \cap\overline{X_{k-1}}\cap X_k]\\
 &\leq E[(z +1)^2] \leq E[(y'+1)^2].
 \end{align*}
\end{proof}

\section{Complexity Sketch}\apndlab{shallow}
Here we give a very brief description of why the order-$k$ sequence Voronoi diagram has $O(nk^3)$ worst-case complexity.
First, recall that by standard lifting, the regular order-$k$ diagram is described by the exact $k$th level in the arrangement of hyperplanes tangent to the unit paraboloid. Since we care about the ordering of the $k$ sites, we are instead concerned with the at most $k$ level. There is a shallow cutting covering the at most $k$ level with $O(n/k)$ vertical prisms each intersecting $O(k)$ planes~\cite{hks-akl-16}. Each prism projects to a triangle in the plane, and thus within this triangle only $O(k)$ sites (corresponding to the planes intersecting the prism) are relevant. Note that $O(k)$ sites can define at most $O(k^4)$ different orderings as they define $O(k^2)$ bisectors and the the arrangment of these bisectors has $O(k^4)$ complexity. Thus the plane is covered by $O(n/k)$ triangles within which the order-$k$ sequence Voronoi diagram has $O(k^4)$ complexity, and thus in total the complexity is $O(nk^3)$.

We remark that for the purposes of this paper an $O(nk^5)$ bound would have sufficed, and is trivial to obtain.  Namely, the worst case complexity of the regular order-$k$ diagram is $O(nk)$, and by the same argument, within each cell there can be at most $O(k^4)$ orderings.

\end{document}